\documentclass[aps,twocolumn,amsmath,amssymb,nofootinbib,superscriptaddress]{revtex4-1}
\usepackage{amsfonts,amssymb,amsmath,amsthm,mathrsfs}
\usepackage{bbold,mathtools}

\usepackage[pdftex]{graphicx}
\usepackage{mathrsfs}
\usepackage[colorlinks]{hyperref}

\newcommand{\bra}[1]{\langle#1|}
\newcommand{\ket}[1]{|#1\rangle}

\newtheorem{theorem}{Theorem}
\newtheorem{lemma}[theorem]{Lemma}

\DeclareMathOperator{\tr}{Tr}

\def\Pr{{\rm Pr} }
\def\>{\rangle}

\def\<{\langle}

\begin{document}

\title{Homomorphic encryption of linear optics quantum computation on almost arbitrary states of light with asymptotically perfect security}

\author{Yingkai Ouyang}
\email[]{y.ouyang@sheffield.ac.uk}
\homepage{http://www.qmetrology.com}
\affiliation{Department of Physics \& Astronomy, University of Sheffield, Sheffield, S3 7RH, United Kingdom}
\affiliation{Singapore University of Technology and Design, 8 Somapah Road, Singapore 487372}
\affiliation{Centre for Quantum Technologies, National University of Singapore, 3 Science Drive 2, Singapore 117543}

\author{Si-Hui Tan}
\affiliation{Singapore University of Technology and Design, 8 Somapah Road, Singapore 487372}
\affiliation{Centre for Quantum Technologies, National University of Singapore, 3 Science Drive 2, Singapore 117543}
\affiliation{Horizon Quantum Computing, 79 Ayer Rajah Crescent, Singapore 139955}

\author{Joseph Fitzsimons}
\affiliation{Singapore University of Technology and Design, 8 Somapah Road, Singapore 487372}
\affiliation{Centre for Quantum Technologies, National University of Singapore, 3 Science Drive 2, Singapore 117543}
\affiliation{Horizon Quantum Computing, 79 Ayer Rajah Crescent, Singapore 139955}

\author{Peter P. Rohde}
\email[]{dr.rohde@gmail.com}
\homepage{http://www.peterrohde.org}
\affiliation{Centre for Quantum Software \& Information (QSI), Faculty of Engineering \& Information Technology, University of Technology Sydney, NSW 2007, Australia}
\affiliation{Hearne Institute for Theoretical Physics and Department of Physics \& Astronomy, Louisiana State University, Baton Rouge, LA 70803, United States}

\date{\today}

\frenchspacing

\begin{abstract}
Future quantum computers are likely to be expensive and affordable outright by few, motivating client/server models for outsourced computation. However, the applications for quantum computing will often involve sensitive data, and the client would like to keep her data secret, both from eavesdroppers and the server itself. Homomorphic encryption is an approach for encrypted, outsourced quantum computation, where the client's data remains secret, even during execution of the computation. We present a scheme for the homomorphic encryption of arbitrary quantum states of light with no more than a fixed number of photons, under the evolution of both passive and adaptive linear optics, the latter of which is universal for quantum computation. The scheme uses random coherent displacements in phase-space to obfuscate client data. In the limit of large coherent displacements, the protocol exhibits asymptotically perfect information-theoretic secrecy. The experimental requirements are modest, and easily implementable using present-day technology.
\end{abstract}

\maketitle

\section{Introduction}

In the upcoming quantum era, it is to be expected that client/server models for quantum computing will emerge, owing to the high expected cost of quantum hardware. This necessitates the ability for a client (Alice), possessing data she wants processed, to outsource the computation to a host (Bob), who possesses the costly quantum computer. In such a model, security will be a major concern. The types of applications to which quantum computing will initially be most relevant will contain sensitive data, whether it be strategically important information, or valuable intellectual property, or confidential personal information. This raises the important question of how Alice can outsource computation of her data such that no adversary Eve, or even the server Bob, can read her data -- she trusts no one!

Homomorphic encryption is a cryptographic protocol that achieves this objective. Alice sends encrypted data to Bob, who processes it in encrypted form, before returning it to Alice. The essential feature is that computing the data does not require first decrypting it -- it remains encrypted throughout the computation, ensuring that even if Bob is compromised, Alice retains integrity of her data.

Classical homomorphic encryption has only been described very recently \cite{Gen09,DGHV10,GHS12}, and a number of results for homomorphic quantum computation have been described \cite{PhysRevLett.109.150501,BJe15,ouyang2015quantum,tan2016quantum,tan2017practical,lai2017statistically,DSS16,alagic2017quantum}. In the case of universal quantum computation, such protocols require a degree of interaction between Alice and Bob. However, it was shown in \cite{PhysRevLett.109.150501} that under certain restricted, non-universal models for quantum computation, homomorphic encryption may be implemented passively, without any client/server interaction, and requiring only separable, non-entangling encoding/decoding operations. In that protocol, in which single photons encode data, random polarisation rotations on Alice's input photonic state obfuscate data from Bob. And in \cite{tan2017practical}, a similar protocol was presented using phase-key encoding, whereby random rotations in phase-space obfuscate Alice's data, encoded into coherent states.

These two protocols are limited in their security by the fact that the rotations in phase-/polarisation-space are correlated across all inputs, thereby limiting the entropy of the encoded input states, and hence its security. For example, with $m$ optical modes, polarisation-key encoding is only able to hide \mbox{$O(\mathrm{log}(m))$} bits of information, falling far short of our utopian ideal of perfect information theoretic security (i.e hiding all $m$ bits of information in the case of 0 or 1 photons per mode).

The polarisation- and phase-key homomorphic encryption techniques are specific examples of a more general framework for encryption, whereby the encoding and decoding operations commute with the computation, thereby mitigating the need for elaborate interactive protocols.

Here we consider an alternate technique that supersedes both polarisation- and phase-key encoding -- \textit{displacement key encoding}, whereby random coherent displacements obfuscate optically-encoded quantum information.
This idea has been recently explored by Marshall {\em et al.} \cite{weedbrook-encryption-2016},
where it was argued heuristically why the scheme might be secure. Based on experimental data generated, Marshall {\em et al.} numerically showed that the mutual information between the encrypted and the unencrypted data can be made small as the variance of the random displacements increases. This encouraging evidence suggests that a displacement key encoding might offer perfect security in the asymptotic limit. 
However, obtaining analytical bounds to quantify the security of the scheme has been recognized to be a challenging issue, yet to be solved.

In this paper, we rigorously obtain explicit bounds on the security of using a displacement key encoding, thereby confirming the intuition of Ref.~\cite{weedbrook-encryption-2016}.
Moreover, the displacement key encoding improves on the earlier polarisation- and phase-key techniques in two important respects. 
First, we demonstrate that by choosing the encoding displacement operators to be independent on each optical mode and to follow a Gaussian distribution with an increasing variance, 
any pair of encoded codewords will become increasingly close in trace-distance and thereby increasingly indistinguishable. 
 Our encoding scheme is a weak information-theoretic security encryption scheme with secrecy error that is twice of this maximum trace distance, and this security definition has been introduced in \cite[Definition 5]{lai2018generalized}. 
We also remark that the trace-distance metric we use is preferable to the mutual information used in Ref.~\cite{weedbrook-encryption-2016}, because the trace distance directly quantifies the indistinguishability of quantum states while the mutual information does not.
Second, our technique is applicable to linear optics computations acting on quantum states of light with no more than a fixed number of photons.
Constraining quantum states to have no more than a fixed number of photons is reasonable, because quantum states that are bounded in energy can always be well approximated by quantum states that bounded in photon number, given that sufficiently many photons is considered.
 This is far more general than polarisation-key encoding, which applies to single-photon input states, or phase-key encoding, which applies to input coherent states.

\section{Commutative homomorphic encryption of passive linear optics} 
 A linear optics network \cite{RevModPhys.79.135}, comprising only beamsplitters and phase-shifters, implements a photon-number-preserving unitary map on the photonic creation operators,
\begin{align}
	\hat{U} \hat{a}_i^\dag \hat{U}^\dag \to \sum_{j=1}^m U_{i,j} \hat{a}_j^\dag,
\end{align}
where $\hat{a}^\dag_i$ is the creation operator for the $i$th mode, there are $m$ optical modes, and $U$ is an $\mathrm{SU}(m)$ matrix characterising the linear optics network.

Bob possesses both the hardware and software for implementing the computation ($\hat{U}$), 	which Alice would like applied to her input state ($\ket\psi_\text{in}$), yielding the computed output state (\mbox{$\ket\psi_\text{out} = \hat{U} \ket\psi_\text{in}	$}). 

Before sending her input state to Bob, Alice, who has limited quantum resources, wishes to encode her input state using operations separable across all modes, similarly for decoding, i.e we rule out entangling gates for Alice. To achieve this, we require the commutation relation,
\begin{align} \label{eq:enc_dec_comm}
	\hat{U} \left[\bigotimes_{i=1}^m \hat{E}_i(k) \right] = \left[\bigotimes_{i=1}^m {\hat{E}'_i(k)} \right] \hat{U},
\end{align}
to hold, where $\hat{E}_i(k)$ ($\hat{E}_i'(k)$) is the encoding (decoding) operation, with key $k$. 
Since Alice has limited classical computational power, she should determine the encoding/decoding operations efficiently with a classical computer, and implement these operations efficiently.
The model is summarised in Fig.~\ref{fig:encryption_model}.
\begin{figure}[!htb]
	\includegraphics[width=0.8\columnwidth]{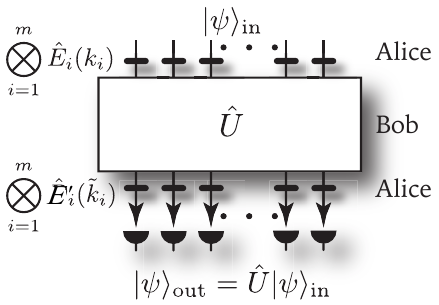}
	\caption{General protocol for commuting homomorphic encryption of optical states under linear optics evolution $\hat{U}$, where $\hat{E}_i$ ($\hat{E}_i'$) are the encoding (decoding) operations, which we require to be separable.} \label{fig:encryption_model}
\end{figure}

The most natural examples of schemes complying with this model are ones where systems encoding quantum information comprise two subsystems: a primary one in which the computation is taking place; and, a secondary independent one, which does not directly couple with the primary and is unaffected by the computational operations. This allows us to exploit the secondary subsystem (e.g polarisation) to control the entropy of our codewords, without affecting the computation in the primary subsystem (e.g photon-number).

\section{Displacement-key encoding} 
Phase-space displacement operations satisfy the required commutation relation of Eq.~(\ref{eq:enc_dec_comm}). The displacement operation adds coherent amplitude to an optical state, thereby translating it in phase-space. This process is described by the unitary displacement operator, given by,
\begin{align}
	\hat{D}(\alpha) = \mathrm{exp}\left[\alpha\hat{a}^\dag - \alpha^*\hat{a}\right].
\end{align}

Displacement operations are easily experimentally implemented using a low-reflectivity beamsplitter and a coherent state (well approximated by a laser source) \cite[Eq (9.15)]{kok2010introduction}, of the form,
\begin{align}
	\ket\alpha = e^{-\frac{|\alpha|^2}{2}} \sum_{n=0}^\infty \frac{\alpha^n}{\sqrt{n!}}	\ket{n},
\end{align}
(see Fig.~\ref{fig:disp_circ}). The displacement amplitude is directly proportional to the coherent state amplitude and the beamsplitter reflectivity. A special case of displaced states are displaced vacuum states, which are identically coherent states of the same amplitude, \mbox{$\hat{D}(\alpha)\ket{0}=\ket\alpha$}.

\begin{figure}[!htb]
\includegraphics[width=0.5\columnwidth]{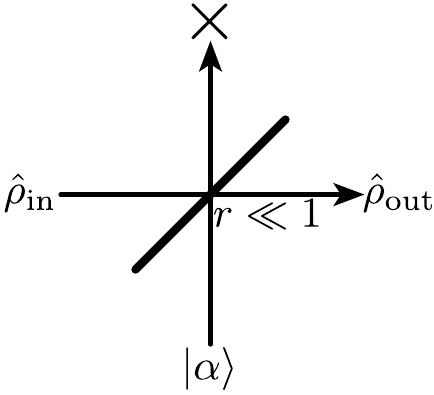}
\caption{Experimental realisation of the displacement operator. A strong coherent state ($\ket\alpha$) is incident on an extremely low reflectivity ($r$) beamsplitter, where it is mixed with the input state ($\hat\rho_\mathrm{in}$). The output state ($\hat\rho_\mathrm{out}$) is now given by the input state, displaced by amplitude $r\alpha$.}\label{fig:disp_circ}
\end{figure}

The commutation relation between displacement operators and linear optics evolution 
relates the output displacement amplitudes $\vec\beta = (\beta_1, \dots, \beta_m)$ to the input displacement amplitudes $\vec\alpha = (\alpha_1 , \dots, \alpha_m)$, and is given by
\begin{align}
	\hat{U} \hat D(\vec \alpha) = \hat D(\vec \beta) \hat{U},	
\end{align}
where $\hat D(\vec \beta) = \bigotimes_{j=1}^m \hat{D}(\beta_j)$, 
$\hat D(\vec \alpha) = \bigotimes_{j=1}^m \hat{D}(\beta_j)$, 
and $\vec\beta$ relates to $\vec\alpha$ according to the unitary map
\begin{align}
	\vec\beta = U\cdot\vec\alpha.
\end{align}
The computation required for Alice to determine her decoding operations from her encoding operations is simple matrix multiplication, which is efficiently computable \cite{arora2009computational}.
Thus, our condition on the complexity of encoding/decoding is satisfied.

An input tensor product of displacement operations with amplitudes $\vec\alpha$ on multiple modes may be be reversed by applying inverse displacement operations with amplitudes $\vec\beta$ at the output, \mbox{$\hat{D}(\vec \beta)^\dag=\hat{D}(-\vec \beta)$}. Specifically,
\begin{align}
	\hat{D}(-\vec\beta) \hat{U} \hat{D}(\vec\alpha) = \hat{U},	
\end{align}
allowing the computation, $\hat{U}\ket{\psi_\mathrm{in}}$, to be recovered from the encoded computation, \mbox{$\hat{U}\hat{D}(\vec\alpha)\ket{\psi_\mathrm{in}}$}, via application of the inverse of the encoding operation.

Our scheme extends trivially to the case where the server is asked to perform any Gaussian operation, rather than only passive linear optical evolution. 
This is because displacements similarly commute with squeezing as one can see from  
\begin{align}
\hat D(\alpha) \hat S(r e^{i \theta}) = \hat S(r e^{i \theta})  \hat D(\gamma) ,
\end{align}
and all Gaussian operations can be expressed as
linear-squeezing-linear evolutions according to the Bloch-Messiah
decomposition \cite{blochmessiah,PhysRevA.71.055801-braunstein-squeezing}, together with displacements, where $\hat S( r e^{i \theta} ) $ denotes a squeezing operator with $r\ge 0$, $\theta \in \mathbb R$ and $\gamma = \alpha \cosh r + \alpha^* e^{i \theta} \sinh r$.

The decryption circuit that Alice uses is identical in structure to her encryption operation, and Alice does not need to able to perform arbitrary linear optical operations that potentially requires up to $m(m-1)/2$ beamsplitters.
Rather, Alice's decryption circuit on $m$ modes always requires only $m$ beamsplitters.
Because of this, Alice's decryption circuit has exactly the same structure as her encryption circuit.
Both the encryption and decryption circuits can then in principle be implemented using $m$ Mach-Zehnder interferometers, and such an optical circuit is independent of Bob's LOQC. 
To find out what coherent states to input into the beamsplitters for the decryption, Alice needs only to know (1) her own secret encrypting displacements, and (2) the unitary that Bob's linear optical circuit implements.

Unlike phase-key or polarisation-key encoding, where the encoding operations applied to each mode must be identical for the encryption/decryption commutation relation to hold, for displacements the amplitudes may be chosen independently for each mode, while still preserving the desired commutation relation. Intuitively, one would anticipate that the ability to choose keys independently for each mode would improve security, since the elimination of correlations between input encoding operations allows the entropy of the encoded state to be greatly increased, thereby making codewords less distinguishable.

We examine this protocol in the context of input data comprising of arbitrary pure quantum states of light
with no more than $n$ photons. In the photon-number basis this implies that,
\begin{align}
\label{eq:input_state_DV}
\ket{\psi_\mathrm{in}} = \sum_{j=0}^{n} \lambda_j \ket{j},
\quad  \hat\rho_\mathrm{in} = 
|\psi_{\mathrm{in}}\> \<\psi_{\mathrm{in}}|,
\end{align}
where $\ket{j} = \frac{1}{\sqrt{j!}}(\hat{a}^\dag)^j \ket{0}$ is a photon-number (Fock) state
and $\hat{a}^\dag$ is the photonic creation operator, and $\ket{\psi_\mathrm{in}}$ has unit norm so that $\sum_{j =0}^\infty |\lambda_j|^2 = 1.$

We consider states supported on no more than $n$ photons, because such states can well approximate states of bounded energy in the following sense.
\begin{lemma} \label{lem:boundedenergy-to-boundedphoton}
Let $\rho = \sum_{j} p_j |\phi_j\>\<\phi_j|$ be a density operator where every $|\phi_j\>$ has expected energy at most $\mu$.
Let $n \ge \mu$. Then there exists a density operator $\rho' = \sum_{j} p_j |\phi'_j\>\<\phi'_j|$ where $|\phi_j\>$ has at most $n$ photons and expected energy at most $\mu$ for every $j$, such that
\begin{align}
\| \rho- \rho' \|_1 \le 4 \sqrt{ \mu/n} + 4 \mu/n.
\end{align}
\end{lemma} 
One can see that the approximation error becomes small when $n$ becomes large for fixed $\mu$.
The proof of Lemma \ref{lem:boundedenergy-to-boundedphoton} follows trivially from Lemma \ref{lem:tracenorm-to-euclideannorm} and Lemma \ref{lem:bounded-support-approximant} in the appendix.  Lemma \ref{lem:tracenorm-to-euclideannorm} and Lemma \ref{lem:bounded-support-approximant} show respectively that the trace-distance between states can be related to the Euclidean norms, and the approximation error for pure states can be bounded using a connection with Markov's inequality. 

\section{Security proof} 
The main result of our paper is the following theorem,
which implies that our encoding scheme in the limit of large coherent displacements has weak information-theoretic security.
\begin{theorem}\label{thm:mainresultofpaper}
The trace-distance between arbitrary encrypted states with at most $n$ photons is at most $\epsilon$, where
\begin{align}
\epsilon =
\frac{1}{\sigma^2} \left(	\frac n 4 
+ 
\frac{1}{2}\left( 1 + \frac{1}{2\sigma^2}\right)^n +
  4(n+1) \right)\notag.
\end{align}
Our scheme thus is a weak information-theoretic security encryption scheme with secrecy error at most $2\epsilon$.
\end{theorem}
Our proof employs a continuous-variable (CV) representation for optical states \cite{braunstein2012quantum}. We omit some intermediate mathematical steps in the main text, delegating the complete step-by-step derivation to the appendix. 

Photon-number (Fock) states are related to the $x$ and $p$ quadrature CVs using Hermite functions \cite[Section 18.1]{arfken2013mathematical}. The Hermite polynomials are defined as,
\begin{align}
	 H_j(x)  = (-1)^n e^{x^2} \frac{d^j}{dx^j} e^{-x^2},  
\end{align}
and the corresponding Hermite functions as,
\begin{align}
	  \psi_j(x) = e^{-x^2/2}\frac{1}{\sqrt{  2^j j! \sqrt \pi  } } H_j(x). 
\end{align}
These provide the direct relation between discrete variable (DV) and CV representations of optical states. Most importantly, for Fock states we have,
\begin{align}
 	|j\rangle   = \int_{-\infty}^{\infty}\psi_j(x) |x\rangle \,dx,  
\end{align}
where $x$ is a position eigenstate in phase-space. The position eigenstates form a complete basis, satisfying,
\begin{align}
	\langle x_1|x_2\rangle = \delta(x_1-x_2).
\end{align} 
Our input state from Eq.~(\ref{eq:input_state_DV}) can therefore be expressed in the position basis as
  \begin{align}
	\hat\rho_\mathrm{in} &=
	 \!\!\!   \sum_{0\le i,j\le n} \lambda_i  \lambda^*_{j} \int_{x_1,x_2\in \mathbb R} \!\!\!  \psi_i(x_1)\psi^*_{j}(x_2) \ket{x_1}\bra{x_2}\, dx_1 dx_2.  
\end{align}  

Let Alice's encoding operation be represented by the quantum process $\mathcal{E}_\mathrm{enc}$, which applies a random complex-valued displacement, chosen from a normal distribution with zero mean and standard deviation $\sigma$. Experimentally, $\sigma$ is bounded by the energy output of coherent laser sources. An unknown encoding operation can be represented as a quantum process,
\begin{align} 
\mathcal{E}_\mathrm{enc}(\hat\rho) =  
\int_{\alpha \in \mathbb C} 
\mu(\alpha)
\hat{D}(\alpha)\hat\rho\hat{D}(\alpha)^\dag 	   d^2 \alpha ,
\end{align}
 where $ 	\mu(\alpha) = \frac{e^{-|\alpha|^2/(2\sigma^2)}}{2 \pi \sigma^2}$
is a Gaussian measure and $d^2 \alpha = d(\Re(\alpha))d(\Im(\alpha))$ indicates that the integral is performed over the real and imaginary parts of $\alpha$. 
 Then our encrypted state $\hat\rho_\mathrm{enc} =\mathcal E_\mathrm{enc} (\rho_{\mathrm{in}})  $ can be interpreted as a weighted mixture over all possible displacement amplitudes associated with the entire key-space. 
 Displacing a position eigenstate by $\alpha = u+i v$ shifts its position by $v$ and appends a phase that depends on its position, $u$ and $v$.
 After performing the integral over the imaginary part of the complex number $\alpha = u+i v$, we get
\begin{align}\label{eq:rho_enc}
\hat\rho_\mathrm{enc} 
&=\!\!   \sum_{0\le i,j\le n} \!\!\!
\lambda_i^*  \lambda_{j}  
\int_{-\infty}^{\infty}
\int_{-\infty}^{\infty}
\int_{-\infty}^{\infty}
\frac{  e^{-u^2/(4\sigma^2)} }{2\sqrt{\pi} \sigma^2}
\psi_i(x_1)\psi_{j}^*(x_2) \nonumber \\
&\times \ket{x_1+u}\bra{x_2+u} \,du\, dx_1\, dx_2.
\end{align} 
The security of the scheme can be quantified using the trace-distance between any pair of its encrypted inputs. 
When the trace-distance between a pair of states in an encryption scheme approaches zero, 
the resolution of this pair of states as perceived by Eve or Bob vanishes.
Such a scheme is said to exhibit weak information-theoretic security \cite{lai2018generalized},
and we proceed to show that our encryption scheme indeed exhibits such a form of security.

To show that the trace-distance between almost arbitrary input states with no more than a fixed number of photons approaches zero as the standard deviation of the random displacements grows, we require detailed information of every 
matrix element $\<a| \hat\rho_{\mathrm{enc}}|b\>$. 
To get a handle on $\<a| \hat\rho_{\mathrm{enc}}|b\>$,
it suffices to consider 
$\<a| \hat\rho_{i,j}|b\>$
where $\hat \rho_{i,j}= \mathcal E(|i\>\<j|)$ because 
$\hat\rho_{\mathrm{enc}} =
\sum_{0\le i,j\le n} 
\lambda_i^*  \lambda_{j}\hat \rho_{i,j} $.
Since $\<a|$ and $|b\>$ can be both expressed in terms of Hermite polynomials in the position basis,
we find that $\<a|\hat \rho_{i,j}|b\>$ is just an integral of the product of four Hermite polynomials.
To evaluate these integrals, we recall that 
any Hermite polynomial $H_j(x)$ can be expressed as the coefficient of $t^j$ in the Gaussian generating function $e^{-x^2/2+2xt-t^2}e^{-x^2/2}j!$ \cite[Eq.  18.5]{arfken2013mathematical}.
Hence, $\<a|\hat \rho_{i,j}|b\>$ may be evaluated by writing all of the Hermite polynomials in terms of their Gaussian generating functions, performing the Gaussian integrals, and then reading off the respective coefficients.
In doing so, we find the exact form of $\<a|\hat \rho_{i,j}|b\>$ in Lemma \ref{lem:exact-form} of the appendix. Namely, $\<a|\hat \rho_{i,j}|b\>$ is only non-zero when $b-a = j-i$. Moreover, we have that
\begin{align}
\<a| \hat \rho_{i,i} |a\>
&=  xy
\sum_{q = 0}^{ {\rm min}(a,i) }  
\binom a {q} 
\binom i {q} 
y^{2q} x^{a+i} ,
\end{align} 
and when $k\ge 1$, we find in Lemma \ref{lem:off-diagonal-bound} of the appendix that
\begin{align}
0 &\le 
|\<a| \hat \rho_{i,i+k} |a+k\> |\notag\\
&\le  
xy
\sum_{q = 0}^{ {\rm min}(a,i) }    
\binom {a+k} {q+k} 
\binom {i+k} {q+k} 
y^{2q+k} x^{a+i+k} ,
\label{eq:off-diagonals}
\end{align} 
where $y=\frac{1}{2\sigma^2}$ and $x  =\frac{2\sigma^2}{1+2\sigma^2}$.
Now let $T$ denote the difference between two encrypted inputs. 
Let us write $T = D+O$ in the Fock basis, where $D$ is the diagonal of $T$. From this decomposition of $T$, we will obtain an upper bound on the trace norm of $T$.
First we prove that the trace norm of $D$ is $O(\sigma^{-2})$.
To see this, we show in Lemma \ref{lem:diagonal-difference} of the appendix that
\begin{align}
&\<a| \hat \rho_{i+1,i+1} |a\> - \<a| \hat \rho_{i,i} |a\> \notag\\
=&
-xy  \<a|\hat \rho_{i,i}|a\>
\notag\\
&+
(xy) x^{a+i+1} 
\sum_{k = 1, \dots , {\rm min}\{a,i\}} 
\binom a {k}
\binom {i} {k-1}
y^{2k} .
\end{align}
We can use this fact to show in Lemma \ref{lem:diagonal-trace-distance} of the appendix that 
$\|\hat \rho_{i+1,i+1}-\hat \rho_{i,i}  \|_{1} \le 2^i \sigma^{-2}$ for $\sigma^2 \ge 2$, from which it follows from a telescoping sum that trace-distance between any pair of encrypted Fock states is at most $2^{n-1} n \sigma^{-2}$.
Next, we upper bound the trace norm of $O$.
To see this, note that the Gersgorin circle theorem \cite{varga-GCT} implies that 
$\| O\|_1$ is at most the sum of the absolute values of all its matrix elements.
By applying a summation of Eq.~(\ref{eq:off-diagonals}) over the indices $a$ and $k$ and by doing the summation in $a$ first, we can use simple binomial identities to find that
$\|O \|_1 \le 8(n+1)\sigma^{-2}$.
Together, with the triangle inequality on the trace norm of $D+O$, 
this allows us to show that the trace-distance between arbitrary encrypted states with at most $n$ photons is at most,
\begin{align}
\frac{1}{\sigma^2} \left(	\frac n 4 
+ 
\frac{1}{2}\left( 1 + \frac{1}{2\sigma^2}\right)^n +
  4(n+1) \right),
	\label{eq:trace-distance}
\end{align}
which asymptotes to zero for large maximum coherent amplitudes in the encoding operations.
This thereby proves Theorem \ref{thm:mainresultofpaper}.

When the client Alice has as her input to the scheme a separable state on $m$ modes, where each mode has at most $n$ photons, it is easy to see using a telescoping bound on the modes that the trace-distance between arbitrary multi-mode separable states is at most $m$ times of the value in \eqref{eq:trace-distance}.
Coherent states with mean photon number of up to $10^8$
can be easily generated in a cavity mode of a pumped laser \cite[Section 4.1]{hanamura2007quantum}.
Since the intensity of a laser can be attenuated with an variable attenuator, this corresponds to having $|\alpha|$ value that ranges between 0 and $10^4$, 
which allows one to create random displacements with $\sigma = 10^4$.
If each mode has at most 15 photons, then using (19), we find that the trace-distance between arbitrary encrypted states on a single mode is at most $6.8\times 10^{-7}$.

\section{Adaptive linear optics} 
Thus far, we have exclusively considered passive linear optics, where there is no measurement or feedforward. However, feedforward -- the ability to measure a subset of the optical modes, and use the measurement outcome to dynamically control the subsequent linear optics network -- is an essential ingredient in many linear optics quantum information processing protocols. For example, when employing single-photon encodings for qubits, it is well known that universal quantum computing is possible with the addition of fast-feedforward \cite{KLM2001}, which is known to require non-linearity \cite{PhysRevA.65.042304}. On the other hand, it is strongly believed that without non-linearity such as feedforward, such schemes cannot be made universal \cite{PhysRevA.65.042304}.

Let us understand intuitively how feedforward and non-linearities can enable 
two different notions of universality in quantum optical computing.
The first notion is CV universality \cite{braunstein2012quantum}, where Braunstein and Lloyd show using Baker-Campbell-Hausdorff arguments how one can in principle implement Hamiltonian evolutions that are arbitrary polynomials of quadrature operators. To achieve this notion of CV universality, it suffices to implement Gaussian unitaries which our scheme can handle natively, along with any non-Gaussian operation which can be achieved using non-linearities.
The second notion of universality is involves DV encoded within CV states, and achieving universal DV quantum computation.
In this notion of DV universality with CV states, non-linearities can help to initialize non-Gaussian states, which are resource states to be consumed during gate teleportation to produce a non-Gaussian gates.
To perform the gate teleportation, one entangles the resource state with a target mode where the non-Gaussian gate is to be computed, and subsequently measures the resource state. 
One then applies a Gaussian gate on the target mode,
conditioned on the measurement outcome.
For instance, on a GKP encoding \cite{GKP01}, a combination of non-Gaussian gates with Gaussian gates can be universal, and such gates can be achieved with feedforward operations with non-linearities.

Can we accommodate for fast-feedforward in the displacement-key homomorphic encryption protocol? Yes we can. Without loss of generality, let us imagine that we wish to measure just one mode and feedforward the measurement outcome to a subsequent round of linear optics, to be once again executed by Bob. For server Bob to perform this measurement, he would have to know the appropriate decryption operator for that mode. However, he does not have this by virtue of the protocol, and Alice cannot provide it to him, lest he misuses it to compromise security.

The only avenue to accommodating the feedforward is to make the protocol interactive. That is, whenever Bob requires a measurement result, to proceed with the computation he outsources the measurement of that mode back to Alice, who returns to him a classical result. This doesn't undermine the viability of the protocol, since Alice is already assumed to have the ability to apply decoding operations, which are by definition separable and can therefore be performed on a per-mode basis.

It is clear that any computation requiring feedforward will necessarily require turning the encryption protocol into an interactive one between Alice and Bob. While this is undesirable, it is to be expected given that no-go proofs have been provided against universal, non-interactive, fully homomorphic protocols \cite{YPF14,newman2017limitations,lai2017statistically}.

\section{Robustness} 
One might wonder how the robustness of our displacement-key encoding scheme to noise compares with the robustness of phase-key and polarization key encoding schemes. In short, because the demands on the structure of the input states of the client Alice is relatively mild, 
she can use bosonic quantum codes on a single mode \cite{GKP01,BinomialCodes2016}. If Alice uses GKP states \cite{GKP01}, so that small imperfections in displacements can be be corrected while the large random displacements can still obfuscate her data from Bob. 
To constrain the photon number per mode, one can use approximate versions \cite{approximate-gkp-1910.08301} of GKP states.
In contrast, bosonic quantum coding schemes are not immediately compatible with the previous phase-key \cite{tan2018practical} and polarization-key schemes \cite{rohde-encrypting-quantum-walk-PhysRevLett.109.150501}.  
For the polarization-key encoding which encrypts boson sampling, without quantum error correction, simulating boson sampling classically remains classically hard with very little noise \cite{bs-error-tolerance-PhysRevA.85.022332} but becomes classically simulable when there is too much noise \cite{bs-error-classical-PhysRevX.6.021039}.
The phase-key scheme \cite{tan2018practical} is only robust to loss errors when the computed states remains entirely classical, and become vulnerable to loss errors once they become entangled into cat states.

\section{Conclusion} We have presented a technique for homomorphic encryption of almost arbitrary optical states under the evolution of linear optics. The scheme requires only separable displacement operations for encoding and decoding, yet provides perfect secrecy in the limit of large displacement amplitudes. For passive linear optics, the protocol requires no client/server interaction, remaining entirely passive. For adaptive linear optics, an interactive protocol is required. The technology for implementing the encoding scheme is readily available today, making near-term demonstration of elementary encrypted optical quantum computation viable.

\section{Acknowledgments} 
Y.O. thanks Jake Iles-Smith for insightful discussions. P.P.R is funded by an ARC Future Fellowship (project FT160100397). 
This research was supported in part by the Singapore National Research Foundation under NRF Award No. NRF-NRFF2013-01.
ST acknowledges support from the Air Force Office of Scientific Research under AOARD grant FA2386-15-1-4082 and FA2386-18-1-4003. ST conducted part of this writing while she was a Guest Researcher at the Niels Bohr International Academy.
Y.O. acknowledges support from Singapore's Ministry of Education, and the US Air Force Office of Scientific Research under AOARD grant FA2386-18-1-4003.

\bibliography{paper}

\appendix

 \begin{widetext}


  \section{Preliminaries}

\subsection{Hermite polynomials}
	    
        Define the Hermite polynomials as
  \begin{align}
H_n(x)  = (-1)^n e^{x^2} \frac{d^n}{dx^n} e^{-x^2},
\end{align}
and the corresponding Hermite functions as 
\begin{align}
\psi_n(x) = e^{-x^2/2}\frac{1}{\sqrt{  2^n n! \sqrt \pi  } } H_n(x).
\end{align}

\subsection{The action of a displacement operator on a position eigenstate}
The displacement operator can be written as
\begin{align}
D(\alpha) = \exp(\alpha \hat a ^\dagger - \alpha^* \hat a),
\end{align}
where $\alpha= u+iv$ is a complex number, with $u,v \in \mathbb R$.
Now the position and momentum operators 
which admit representations as $x$ and $\frac{1}{i} \frac{d}{dx}$ respectively 
can also be written as dimensionless quadratures $X_1$ and $X_2$ respectively 
which can be related to the ladder operators via the equalities
\begin{align}
a  &= \frac{1}{\sqrt 2} (\hat X_2 - i \hat X_1) , \quad
a ^\dagger = \frac{1}{\sqrt 2} (\hat X_2 + i \hat X_1) ,
\end{align}
which implies that 
$\hat X_1 =  \frac{1}{\sqrt 2} (\hat a ^\dagger - \hat a) $
 and 
$\hat X_2 = \frac{1}{i\sqrt 2} \left(\hat a ^\dagger - \hat a \right)$ respectively.
Then 
\begin{align}
[\hat X_1, \hat X_2]
&=
\frac{1}{2i}[\hat a ^\dagger+ \hat a , \hat a ^\dagger - \hat a ] \notag\\
&=
\frac{1}{2i}
\left(
[\hat a ^\dagger, \hat a ^\dagger - \hat a ] 
+
[\hat a , \hat a ^\dagger - \hat a ] 
\right)
\notag\\
&=
\frac{1}{2i}
\left(
[\hat a ^\dagger, - \hat a ] 
+
[\hat a ,  \hat a ^\dagger ] 
\right) \notag\\
&=
\frac{-1}{i}
[\hat a ^\dagger, \hat a ] \notag\\
&= i.
\end{align}
Since $[\hat X_1,\hat X_2]=i$ the dimensionless quadrature operators 
$\hat X_1$ and $\hat X_2$ indeed satisfy the canonical commutation relations.
We then write the displacement operator in terms of the quadrature operators to get
\begin{align}
D(\alpha) 
&= \exp(\alpha (\frac{1}{\sqrt 2}\hat X_1 - \frac i {\sqrt  2} \hat X_2) - \alpha^* (\frac 1 {\sqrt 2} \hat X_1 + \frac{i}{\sqrt 2} \hat X_2) ) \notag\\
&= \exp( (\alpha - \alpha^*) \hat X_1 / \sqrt 2 + i \hat X_2(-\alpha - \alpha^* )/\sqrt 2 ) \notag\\
&= \exp( \sqrt 2iv \hat X_1 - i \hat X_2 \sqrt 2u ) \notag\\
&= \exp( i (\sqrt 2 v \hat X_1 -  \sqrt 2 u \hat X_2  ) ) .
\end{align}
Now recall that the BCH formula for operators $A,B$ whose commutator is proportional to the identity operator is
$e^{i(A+B)} = e^{iA}e^{iB}e^{-\frac{i^2}{2}[A,B]}=
e^{iA}e^{iB}e^{[A,B]/2} $.
Hence
\begin{align}
D(\alpha) 
&= e^{ i \sqrt 2 v \hat X_1}e^{- i \sqrt 2 u \hat X_2 }e^{[\sqrt 2v\hat X_1,-\sqrt 2 u \hat X_2]/2} \notag\\
&= e^{ i \sqrt 2 v \hat X_1}e^{- i \sqrt 2 u \hat X_2 }e^{-uv [\hat X_1,\hat X_2]} \notag\\
&= e^{ i \sqrt 2 v \hat X_1}e^{- i \sqrt 2 u \hat X_2 }e^{-i uv } .
\end{align}
Now let $|x\>_1$
 denote an eigenstate of the quadrature operator $\hat X_1$ with eigenvalue $x$,
 so that $\hat X_1 |x\>_1 = x |x\>_1$.
 Then it is clear that $e^{i\theta X_1} |x\>_1 = e^{i \theta x} |x\>_1$.
 The position eigenstate can be written in the momentum basis, which is also its Fourier basis, so 
 \begin{align}
|x\>_1 = \frac{1}{\sqrt{2\pi}}\int_{-\infty}^\infty dp e^{-ipx} |p\>_2,
\end{align}
where $|p\>_2$ denotes an eigenstate of the second quadrature operator $\hat X_2$ with eigenvalues $p$. Hence
\begin{align}
e^{i \theta \hat X_2} |x\>_1
&= \frac{1}{\sqrt{2\pi}}\int_{-\infty}^\infty dp e^{-ipx} e^{i \theta \hat X_2} |p\>_2 \notag\\
&= \frac{1}{\sqrt{2\pi}}\int_{-\infty}^\infty dp e^{-ipx} e^{i \theta p} |p\>_2 \notag\\
&= \frac{1}{\sqrt{2\pi}}\int_{-\infty}^\infty dp e^{-ip(x-\theta)}  |p\>_2 \notag\\
&= |x-\theta\>_1.
\end{align}
Hence it follows that 
\begin{align}
D(u+iv)|x\>_1
 &=
 e^{ i \sqrt 2 v \hat X_1}e^{- i \sqrt  2 u \hat X_2 }e^{-i uv } |x\>_1 \notag\\
&=
 e^{ i \sqrt 2 v \hat X_1}e^{-i uv } |x+\sqrt  2 u\>_1 \notag\\
&=
 e^{ i \sqrt 2 v(x+u)}e^{-i uv } |x+\sqrt 2 u\>_1 \notag\\	
&=
 e^{ i \sqrt 2 v x}e^{i (\sqrt 2 -1 ) uv } |x+\sqrt 2u\>_1 .
\end{align}
 
  \section{Representation of the encrypted state}

\begin{lemma}
\label{lem:rho-representation}
Let $|\psi\rangle  = \sum_{i=0}^n \lambda_i |i\rangle $ for any $\lambda_i \in \mathbb C$ such that $\<\psi|\psi\>=1$.
Let $\mathcal E$ be the encryption operation that randomly displaces with a complex number $u+iv$, where $u$ and $v$ are chosen independently from normal distributions with mean 0 and standard deviation $\sigma$.
Let $\rho_{\mathrm{enc}} = \mathcal E(|\psi\>\<\psi|)$. Then
\begin{align}
\hat \rho_{\mathrm{enc}}
=
\sum_{i=0}^n \sum_{j=0}^n  \lambda_i  \lambda_j^* 
\frac{1}{2 \sqrt \pi \sigma} 
\int_{-\infty}^\infty
du\ 
e^{-\frac{u^2}{4 \sigma^2}} 
\int_{-\infty}^{\infty} dx  \int_{-\infty}^{\infty} dy\ 
e^{- \sigma ^2 (x-y)^2}
\psi_i(x)\psi_j(y) |x+ u\rangle \langle y+u| . 
\end{align}
\end{lemma}
\begin{proof}
Note the Fock states can be written in the position basis, so that for all non-negative integers $i$ we have
  \begin{align}
|i\rangle   = \int_{-\infty}^{\infty} dx  \ \psi_i(x) |x\rangle_1 .
\end{align}
 Then
$|\psi\rangle  \langle \psi|= \sum_{i=0}^n \sum_{j=0}^n  \lambda_i  \lambda_j^* |i\rangle \langle  j|$.
Expanding this out in the position basis, and dropping the labels on the first quadrature eigenstates, we get
\begin{align}
|\psi\rangle  \langle \psi|
= 
\sum_{i=0}^n \sum_{j=0}^n  \lambda_i  \lambda_j^* 
\int_{-\infty}^{\infty} dx_1  \int_{-\infty}^{\infty} dx_2\ 
\psi_i(x_1)\psi_j(x_2) |x_1\rangle \langle  x_2|.
\end{align}
Then for real $u$ and $v$, we get
\begin{align}
D(u+iv) |\psi\rangle  \langle \psi| D(u+iv) ^\dagger
&=
\sum_{i=0}^n \sum_{j=0}^n  \lambda_i  \lambda_j^* 
\int_{-\infty}^{\infty} dx_1  \int_{-\infty}^{\infty} dx_2\ 
e^{i\sqrt 2 v(x_1-x_2)}
\psi_i(x_1)\psi_j(x_2) |x_1+\sqrt 2 u\rangle \langle  x_2+\sqrt 2 u|.
\end{align} 
Encrypting the state $|\psi\>\<\psi|$ and changing the variable with respect to $u$ then gives
\begin{align}
\hat \rho_{\mathrm{enc}}
&= 
\sum_{i=0}^n 
\sum_{j=0}^n  \lambda_i  \lambda_j^*  
\frac{1}{2\pi \sigma^2} 
\int_{-\infty}^\infty
\int_{-\infty}^\infty
du\ dv\ e^{-\frac{u^2 + v^2}{2\sigma^2}} 
\int_{-\infty}^{\infty} dx_1  \int_{-\infty}^{\infty} dx_2\ 
e^{i\sqrt 2 v(x_1-x_2)}
\psi_i(x_1)\psi_j(x_2) |x_1+\sqrt 2 u\rangle \langle x_2+\sqrt 2 u| \notag\\
&= 
\sum_{i=0}^n \sum_{j=0}^n 
\lambda_i  \lambda_j^*  
\frac{1}{4\pi \sigma^2} 
\int_{-\infty}^\infty
\int_{-\infty}^\infty
du\ dv\ e^{-\frac{u^2 + v^2}{4\sigma^2}} 
\int_{-\infty}^{\infty} dx_1  \int_{-\infty}^{\infty} dx_2\ 
e^{i v(x_1-x_2)}
\psi_i(x_1)\psi_j(x_2) |x_1+ u\rangle \langle x_2+u| .
\end{align}
We can perform the integral with respect to $v$ to arrive at 
\begin{align}
\hat \rho_{\mathrm{enc}}
&= 
\sum_{i=0}^n \sum_{j=0}^n  \lambda_i  \lambda_j^* 
\frac{1}{4\pi \sigma^2} 
\int_{-\infty}^\infty
du\ 
e^{-\frac{u^2}{4\sigma^2}} 
\int_{-\infty}^{\infty} dx_1  \int_{-\infty}^{\infty} dx_2\ 
2 \sqrt{\pi } \sigma  e^{- \sigma ^2 (x_1-x_2)^2}
\psi_i(x_1)\psi_j(x_2) |x_1+ u\rangle \langle x_2+u|.
\end{align}
Simplifying the above and relabeling the variables in the integration then gives the result.
\end{proof}

\section{Integrals of products of Hermite polynomials}

The following lemma gives a bound for the exponential suppression of a certain integral of products of Hermite polynomials in the orders of the some of the Hermite polynomials. The key tools used here are generating functions for the Hermite polynomials, and this leads to a significant improvement of bounding the absolute value of the integral of product of Hermite functions over that in Ref.~\cite{ouyang2013truncated}.

Now let us define the integral
\begin{align}
I_{a,b,i,j} = 
\frac{1}{2 \sqrt \pi \sigma} 
\int_{-\infty}^\infty
du\ 
e^{-\frac{u^2}{4 \sigma^2}} 
\int_{-\infty}^{\infty} dx  
\int_{-\infty}^{\infty} dy\ 
e^{- \sigma ^2 (x-y)^2}
\psi_i(x)\psi_j(y) 
\psi_a(x+u)\psi_b(y+u),
\end{align}
so that 
\begin{align}
\<a| \hat \rho_{\rm enc}|b\> = 
\sum_{i=0}^n \sum_{j=0}^n
\lambda_i \lambda_j^* I_{a,b,i,j}.
\end{align}
If we encrypt another state of the form $\sum_{i=0}^n \mu_i |i\>$,
then the difference between the two matrix elements will be
\begin{align}
\sum_{i=0}^n \sum_{j=0}^n
\left(
\lambda_i \lambda_j^* I_{a,b,i,j}
-
\mu_i \mu_j^* I_{a,b,i,j}
\right)
=
\sum_{i=0}^n \sum_{j=0}^n
\left(
\lambda_i \lambda_j^* -\mu_i \mu_j^* 
\right)
I_{a,b,i,j}.
\end{align}
Define $\hat \rho_i = \mathcal E(|i\>\<i|)$, 
and define 
$\hat \rho_{i,j} = \mathcal E(|i\>\<j|)$.
Then
\begin{align}
\<a|\hat \rho_{i,j} |b\>  = I_{a,b,i,j}.
\end{align}
Clearly for $\rho = \sum_{i,j} \lambda_i \lambda_j^* |i\>\<j|$,
by linearity of the encryption operation, 
\begin{align}
\hat \rho  = \mathcal E(\rho) =  \sum_{i,j} \lambda_i \lambda_j^*  \hat \rho_{i,j}.
\end{align}

We use the method of generating functions to evaluate the exact form for the integral $I_{a,b,i,j}$.
\begin{lemma}
\label{lem:exact-form}
 Let $a,b,i,j$ be non-negative integers and $\sigma>0$.  Let $x= \frac{2 \sigma^2}{1+2\sigma^2} $ and $ y=1/(2 \sigma^2)$  .
Then 
\begin{align}
\<a|\hat \rho_{i,j} |b\>  =I_{a,b,i,j} = \frac{1}{1+2\sigma^2}
\sum_{\substack{
i_1,i_2,i_3,i_4 \ge 0\\
i_1+i_2 =a\\
i_1+i_3 =b\\
i_2+i_4 =i\\
i_3+i_4 =j\\
}} 
y^{i_2+i_3}
\sqrt{ 
\binom a {i_2} \binom b {i_3} \binom i {i_2} \binom j {i_3}
} 
\sqrt{x^{a+b+i+j}}. \label{eq:lemma2}
\end{align}  
\end{lemma}
\begin{proof}   Let $I=I_{a,b,i,j}$.
The generating function of the Hermite polynomial is given by
\begin{align}
\exp(-x^2/2 +2xt - t^2 ) = \sum_{n=0}^\infty  e^{-x^2/2} H_n(x) t^n /n!.
\end{align}
Hence, using the notation $[t^n] f(t)$ to denote the coefficient of $t^n$ in an analytical function $f(t)$, we get
\begin{align}
 H_n(x) = [t^n] \exp( 2xt - t^2 )  n! .
\end{align}
Recall that
\begin{align}
\psi_n(x) = e^{-x^2/2}\frac{1}{\sqrt{  2^n n! \sqrt \pi  } } H_n(x).
\end{align}
Now
\begin{align}
&\psi_i(x) \psi_j(y)  \psi_a(x+u) \psi_b(y+u) \notag\\
=
&
e^{-(x^2+y^2)/2}
e^{-((x+u)^2+(y+u)^2)/2}
  \frac{  H_i(x) H_j(y) H_a(x+u) H_b(y+u)  }
  {\pi \sqrt{  2^{i+j+a+b} i!j!a!b!   } } 
  \notag\\
=
&
[s^i t^j f^a g^b]
  \frac{   \sqrt{ i!j!a!b! }  }
  {\pi \sqrt{  2^{i+j+a+b}  } } 
e^{-x^2/2 +2xs - s^2 }
e^{-y^2/2 +2yt - t^2 } 
e^{-(x+u)^2/2 +2(x+u)f - f^2 } 
e^{-(y+u)^2/2 +2(y+u)g - g^2 } .
\end{align}
Hence
\begin{align}
I&= 
[s^i t^j f^a g^b]
  \frac{   \sqrt{ i!j!a!b! }  }
  {\pi \sqrt{  2^{i+j+a+b}  } } 
  \frac{1}{2 \sqrt \pi \sigma} 
\int_{-\infty}^\infty
du\ 
\int_{-\infty}^{\infty} dx  
\int_{-\infty}^{\infty} dy\ 
e^{-\frac{u^2}{4 \sigma^2}} 
e^{- \sigma ^2 (x-y)^2}
e^{-x^2/2 +2xs - s^2 }
e^{-y^2/2 +2yt - t^2 } 
\notag\\
& \times e^{-(x+u)^2/2 +2(x+u)f - f^2 } 
e^{-(y+u)^2/2 +2(y+u)g - g^2 } .
\end{align}
This integral can be easily performed.
We make use of the identity 
\begin{align}
\int_{-\infty}^{\infty} dy\ e^{-a y^2 -by }
=
\sqrt{\frac{\pi }{a}}  e^{\frac{b^2}{4 a}} ,
\end{align}
where $a >0$. 
Using this identity repeatedly, we can show that 
\begin{align}
I=   \frac{   \sqrt{ i!j!a!b! }  }
  { \sqrt{  2^{i+j+a+b}  } } \alpha F
\end{align}
where $\alpha = 1/(1+2\sigma^2)$ and
\begin{align}
F   
&=
[s^i t^j f^a g^b]
\exp\left[ 
\alpha \left(
4 f g \sigma ^2 + 2 f s +2 g t + 4 s t \sigma ^2   
   \right) 
\right].
\end{align}
By writing the exponential in $F$ as a product of four exponentials, and using the Taylor series expansion for each, we have
\begin{align}
F   
&=
[s^i t^j f^a g^b]
e^{ \alpha 4 f g \sigma ^2 }e^{\alpha 2 f s} e^{\alpha 2 g t} e^{\alpha 4 s t \sigma ^2   } \notag\\
&=
[s^i t^j f^a g^b]
\sum_{i_1,i_2,i_3,i_4 \ge 0}
\frac{ (\alpha 4 f g \sigma ^2)^{i_1} }{i_1!}
\frac{(\alpha 2 f s)^{i_2}}{i_2!}
\frac{(\alpha 2 g t)^{i_3}}{i_3!}
\frac{(\alpha 4 s t \sigma ^2 )^{i_4}  }{i_4!}.
\end{align}
By extracting the coefficients, we get
\begin{align}
F   
&=
\sum_{\substack{
i_1,i_2,i_3,i_4 \ge 0\\
i_1+i_2 =a\\
i_1+i_3 =b\\
i_2+i_4 =i\\
i_3+i_4 =j\\
}}
\frac{ (\alpha 4  \sigma ^2)^{i_1} }{i_1!}
\frac{(\alpha 2 )^{i_2}}{i_2!}
\frac{(\alpha 2 )^{i_3}}{i_3!}
\frac{(\alpha 4  \sigma ^2 )^{i_4}  }{i_4!}
\notag\\
&=
\sum_{\substack{
i_1,i_2,i_3,i_4 \ge 0\\
i_1+i_2 =a\\
i_1+i_3 =b\\
i_2+i_4 =i\\
i_3+i_4 =j\\
}}
\frac{ (2  \sigma ^2)^{i_1} }{i_1!}
\frac{(1 )^{i_2}}{i_2!}
\frac{(1 )^{i_3}}{i_3!}
\frac{(2  \sigma ^2 )^{i_4}  }{i_4!}
(2\alpha)^{i_1+i_2+i_3+i_4}
\notag\\
&=
\sum_{\substack{
i_1,i_2,i_3,i_4 \ge 0\\
i_1+i_2 =a\\
i_1+i_3 =b\\
i_2+i_4 =i\\
i_3+i_4 =j\\
}}
\frac{ (2  \sigma ^2)^{i_1+i_4} }
{i_1! i_2! i_3! i_4!}
(2\alpha)^{i_1+i_2+i_3+i_4}.
\end{align}
Clearly, $i_1+i_2+i_3+i_4 = \frac{	a+b+i+j}{2}$.
Thus
\begin{align}
F&=
\sum_{\substack{
i_1,i_2,i_3,i_4 \ge 0\\
i_1+i_2 =a\\
i_1+i_3 =b\\
i_2+i_4 =i\\
i_3+i_4 =j\\
}}
\frac{ (2  \sigma ^2)^{i_1+i_4} }
{i_1! i_2! i_3! i_4!}
\sqrt{(2\alpha)^{a+b+i+j}}
=
\sum_{\substack{
i_1,i_2,i_3,i_4 \ge 0\\
i_1+i_2 =a\\
i_1+i_3 =b\\
i_2+i_4 =i\\
i_3+i_4 =j\\
}}
\frac{ (2  \sigma ^2)^{i_1+i_4} }
{i_1! i_2! i_3! i_4!}
\sqrt{2^{a+b+i+j}}
\sqrt{\frac{1}{(1+2\sigma^2)^{a+b+i+j}}}
\end{align}
for $\alpha = 1/(1+2\sigma^2)$.
Note that
\begin{align}
F =
\sum_{\substack{
i_1,i_2,i_3,i_4 \ge 0\\
i_1+i_2 =a\\
i_1+i_3 =b\\
i_2+i_4 =i\\
i_3+i_4 =j\\
}} (2 \sigma^2)^{-i_2-i_3}
\frac{\sqrt{2^{a+b+i+j}}}{i_1! i_2! i_3! i_4!}
\sqrt{\left( \frac{2 \sigma^2}{1+2\sigma^2} \right)^{a+b+i+j}}.
\end{align}
Now note that $i_1 = a-i_2$, $i_1 = b-i_3$, $i_4 = i - i_2$ and $i_4 = j-i_3$, which implies that
\begin{align}
i_1! i_2! i_3! i_4! = \sqrt{(a-i_2)!(b-i_3)!i_2! i_2! i_3!i_3! (i-i_2)! (j-i_3)! }.
\end{align}
Therefore
\begin{align}
\frac{\sqrt{a!b!i!j!}}{i_1! i_2! i_3! i_4!}
&=
\sqrt{ 
\frac{a!b!i!j!}{(a-i_2)!(b-i_3)!i_2! i_2! i_3!i_3! (i-i_2)! (j-i_3)!}
} \notag\\
&=
\sqrt{ 
\binom a {i_2} \binom b {i_3} \binom i {i_2} \binom j {i_3}
}  .
\end{align}
Making appropriate substitutions then completes the proof.
\end{proof}
  
\section{Towards the proof of the indistinguishability bound}  
The key result that we rely on is the result from Lemma \ref{lem:exact-form} which gives an exact form for $I_{a,b,i,j}$ in terms of $y = 1/(2 \sigma^2)$ and $x = \frac{2\sigma^2}{1+2\sigma^2}$.
Now let $b=a+k$ for $k \ge 0$. 
Then observe that $I_{a,b,i,j}=0$ unless $j = i+k$. Hence we restrict our attention to this case. 
Then we have 
\begin{align}
I_{a,a+k,i,i+k}
&=
\frac{1}{1+2\sigma^2}
\sum_{i_2 = 0, \dots , {\rm min}(a,i)} 
\sqrt{
\binom a {i_2}
\binom {a+k} {i_2+k}
\binom i {i_2}
\binom {i+k} {i_2+k}
} 
y^{2i_2+k} x^{a+i+k} .
\end{align}
To see this, Lemma \ref{lem:exact-form}. Recall that the subscripts for the summation in Eq (\ref{eq:lemma2}) must satisfy the equalities
\begin{align}
i_1+i_2 &= a  \label{eq:(i)} \\
i_1+i_3 &= b  \label{eq:(ii)} \\
i_2+i_4 &= i  \label{eq:(iii)} \\
i_3+i_4 &= j  \label{eq:(iv)} .
\end{align}
We can then get
\begin{align}
(\ref{eq:(ii)})-(\ref{eq:(i)})&:   b-a=i_3-i_2 \\
(\ref{eq:(iv)})-(\ref{eq:(iii)})&: j-i = i_3-i_2.
\end{align}
Hence $b-a = j-i.$
So if $b=a+k$, then $(a+k)-a=j-i$ which implies that $k=j-i$ and hence $j=i+k$. Hence whenever $j \neq i+k$, there will be nothing in the summation of (\ref{eq:lemma2}) to sum over, and the summation in that case evaluates to zero.

Before we proceed, we provide the proofs of several simple but useful technical lemmas.
The first technical lemma we need is the following combinatorial identity.
\begin{lemma}
\label{lem:summation-binomial}
Let $0<x<1$, and let $k$ be a non-negative integer. 
Then $\sum_{a \ge k} x^a \binom a k = \frac{x^k}{(1-x)^{k+1}} $.
\end{lemma}
\begin{proof}
First note that by relabeling the index for the summation,
the sum in the lemma is equal to $\frac{1}{k!} \sum_{a\ge 0} x^{a+k} (a+k)\dots(a+1) 
= \frac{x^k}{k!} \frac{d^k}{dx^k} \sum_{a \ge 0} x^{a+k}$.
By use the generating function $1/(1-x)$ which holds because $|x|<1$, the summation becomes 
$\frac{x^k}{k!} \frac{d^k}{dx^k} \frac{x^k}{1-x}$. Simplifying this using the fact that $1-x^k = (1-x)(1+\dots +x^{k+1})$ yields the result.
\end{proof}
The next technical lemma we need also involves binomial coefficients.
\begin{lemma}
\label{lem:binomial-difference}
Let $x=(2\sigma^2)/(1+2\sigma^2)$, and let $i$ and $k$ be non-negative integers such that $0 \le k \le i-1$. Then
$\binom{i+1}{k}x - \binom i k =  \binom i k\left( \frac{k x}{i-k+1} - \frac{1}{2\sigma^2+1} \right).$
\end{lemma}
\begin{proof}
Note that
$\binom{i+1}{k}x - \binom i k = \frac{k x}{i-k+1} - \frac{1}{2\sigma^2+1} $
is equal to 
$\binom i k\left(
\frac{(i+1)x}{i-k+1} - 1
\right)$.
Next it is easy to see that 
$\frac{i+1}{i-k+1} = 1+ \frac{k}{i-k+1} $.
Hence 
\begin{align}
\binom{i+1}{k}x - \binom i k =
\binom i k\left(
x + x \frac{k}{i-k+1} -1
\right) =
\binom i k\left(
  \frac{-1}{2\sigma^2+1}+   \frac{kx}{i-k+1}
\right) \notag
\end{align}
which proves the result.
\end{proof}

Note the trivial fact that $\sum_{k\ge 0 } x^k = 1+2\sigma^2$.
Let us consider the case of $k=0$ first, which corresponds to $i=j$. Hence we consider the non-zero matrix elements of $\hat \rho_i$, which are $\<a|\hat \rho_i|a\>$ for $a = 0, 1,\dots,$.
Notice then that we have
\begin{align}
\<a|\hat \rho_i|a\>
&= I_{a,a,i,i} \notag\\
&=
\frac{1}{1+2\sigma^2}
\sum_{i_2 = 0, \dots , {\rm min}(a,i)} 
\binom a {i_2}
\binom i {i_2}
y^{2i_2} x^{a+i} .
\end{align}
We are then in a position to bound the trace distance between $\hat \rho_{i+1}$ and $\hat \rho_i$ for every integer $i$. In the lemma that follows, we only consider positive integer $i$, because the case of $i=0$ has already been shown earlier.
\begin{lemma}
\label{lem:diagonal-difference}
Let $i$ and $a$ be any non-negative integer. Let $x=(2\sigma^2)/(1+2\sigma^2)$ and $y = 1/(2\sigma^2)$ for $\sigma > 0$. Then 
\begin{align}
\<a|\hat \rho_{i+1}|a\> -\<a| \hat \rho_i |a\>
&=
\frac{- \<a|\hat \rho_i|a\>}{2\sigma^2+1}
+
\frac{x^{a+i+1} }{1+2\sigma^2}
\sum_{k = 1, \dots , {\rm min}\{a,i\}} 
\binom a {k}
\binom {i} {k-1}
y^{2k}   .
\end{align} 

\end{lemma}
\begin{proof}
To prove this, we consider two scenarios. In one scenario, $a$ is small in the sense that $ a\le i$. In the other scenario, $a > i$. 

When $a \le i$, using Lemma \ref{lem:binomial-difference}, we have the following
\begin{align}
\<a|\hat \rho_{i+1}|a\> - \<a|\hat \rho_i|a\>
&= 
\frac{1}{1+2\sigma^2}
\sum_{i_2 = 0, \dots , a} 
\binom a {i_2}
\binom {i+1} {i_2}
y^{2i_2} x^{a+i+1} 
-
\frac{1}{1+2\sigma^2}
\sum_{i_2 = 0, \dots , a} 
\binom a {i_2}
\binom {i} {i_2}
y^{2i_2} x^{a+i} \notag\\ 
&=
\frac{1}{1+2\sigma^2}
\sum_{i_2 = 0, \dots , a} 
\binom a {i_2}
\binom {i} {i_2}
y^{2i_2} x^{a+i} 
\left(
\frac{i_2 x}{i-i_2+1} - \frac{1}{2\sigma^2+1}
\right).\notag
\end{align}
Using the expansion for $\<a|\hat \rho_i|a\>$, we then get
\begin{align}
\<a|\hat \rho_{i+1}|a\> - \<a|\hat \rho_i|a\>
&=  
\frac{- \<a|\hat \rho_i|a\>}{2\sigma^2+1}
+
\frac{1}{1+2\sigma^2}x^{a+i} 
\sum_{i_2 = 1, \dots , a} 
\binom a {i_2}
\binom {i} {i_2-1}
y^{2i_2} x .\label{lem:case1}
\end{align} 
Now we proceed to consider the case when $a > i$. 
Then we can use Lemma \ref{lem:binomial-difference} again to get
\begin{align}
\<a|\hat \rho_{i+1}|a\> - \<a|\hat \rho_i|a\>
&= 
\frac{1}{1+2\sigma^2}
\sum_{i_2 = 0, \dots , i+1} 
\binom a {i_2}
\binom {i+1} {i_2}
y^{2i_2} x^{a+i+1} 
-
\frac{1}{1+2\sigma^2}
\sum_{i_2 = 0, \dots , i} 
\binom a {i_2}
\binom {i} {i_2}
y^{2i_2} x^{a+i} \notag\\ 
&=
\frac{1}{1+2\sigma^2}
\sum_{i_2 = 0, \dots , i} 
\binom a {i_2}
\binom {i} {i_2}
y^{2i_2} x^{a+i} 
\left(
\frac{i_2 x}{i-i_2+1} - \frac{1}{2\sigma^2+1}
\right)
+\frac{1}{1+2\sigma^2}
\binom a {i+1} y^{2i+2} x^{a+i+1}.\notag
\end{align}
Hence we get 
\begin{align}
\<a|\hat \rho_{i+1}|a\> - \<a|\hat \rho_i|a\>
&=  
\frac{- \<a|\hat \rho_i|a\>}{2\sigma^2+1}
+\frac{1}{1+2\sigma^2}
x^{a+i+1} 
\sum_{i_2 = 1, \dots , i} 
\binom a {i_2}
\binom {i} {i_2-1}
y^{2i_2} 
+\binom a {i+1} y^{2i+2} x^{a+i+1}\notag\\
&= 
\frac{- \<a|\hat \rho_i|a\>}{2\sigma^2+1}
+\frac{1}{1+2\sigma^2}
x^{a+i+1} y^2
\sum_{i_2 = 0, \dots , i} 
\binom a {i_2+1}
\binom {i} {i_2}
y^{2i_2} ,\label{lem:case2}
\end{align}
and the result follows from (\ref{lem:case1}) and (\ref{lem:case2}). 
\end{proof} 
The trace distance between $\hat \rho_{i+1}$ and $\hat \rho_i$ is suppressed with increasing $\sigma$, as we shall now show.
\begin{lemma}
\label{lem:diagonal-trace-distance}
The trace distance between $\hat \rho_{i+1}$ and $\hat \rho_i$ is 
\begin{align}
\frac{1}{2}
\|\hat \rho_{i+1} - \hat \rho_i\|_1 
\le
\frac{1}{2(1+2\sigma^2)} + \frac{1}{4\sigma^2}\left( 1 + \frac{1}{2\sigma^2}\right)^i.\notag
\end{align}
\end{lemma}
\begin{proof}
Since $\hat \rho_{i+1}$ and $\hat \rho_i$ are diagonal matrices in the number basis, we have
\begin{align}
 \|\hat \rho_{i+1} - \hat \rho_i\|_1 
&= \sum_{a \ge 0 }
 \left|
\<a| \rho_{i+1} |a\> - \<a| \hat \rho_i |a\>
\right|.
\end{align}
Using Lemma \ref{lem:diagonal-difference} for the exact form of $\<a| \rho_{i+1} |a\> - \<a| \hat \rho_i |a\>$, we get
\begin{align}
 \|\hat \rho_{i+1} - \hat \rho_i\|_1 
&\le 
\sum_{a \ge 0 }
\frac{  \<a|\hat \rho_i|a\>}{2\sigma^2+1}
+
\sum_{a \ge 0 }
\frac{x^{a+i+1} }{1+2\sigma^2}
\sum_{k = 1, \dots , {\rm min}\{a,i\}} 
\binom a {k}
\binom {i} {k-1}
y^{2k}  \notag\\
&\le 
\sum_{a \ge 0 }
\frac{  \<a|\hat \rho_i|a\>}{2\sigma^2+1}
+
\sum_{a \ge 0 }
\frac{x^{a+i+1} }{1+2\sigma^2}
\sum_{k = 1}^{\infty} 
\binom a {k}
\binom {i} {k-1}
y^{2k}  ,
\end{align}
where $\binom a {k} = 0$ for all $k >a$.
The first summation above is trivial to bound because the trace of a density matrix must be one, so one must have $\sum_{a \ge 0} \<a|\hat \rho_i|a\>= 1$.
For the second summation, we can use Lemma \ref{lem:summation-binomial} 
to get
\begin{align}
 \|\hat \rho_{i+1} - \hat \rho_i\|_1 
&\le  
\frac{1}{2\sigma^2+1}
+ 
\frac{x^{ i+1} }{1+2\sigma^2}
\sum_{k = 1} ^\infty
\frac{x^k}{(1-x)^{k+1}}
\binom {i} {k-1}
y^{2k} .
\end{align}
Since $x/(1-x) = 2\sigma^2=y^{-1}$ and $1/(1-x) = 2\sigma^2+1$ for $x = (2\sigma^2)/(1+2\sigma^2)$, we get
\begin{align}
 \|\hat \rho_{i+1} - \hat \rho_i\|_1 
&\le  
\frac{1}{2\sigma^2+1}
+ 
 x^{ i+1} 
\sum_{k = 1}  ^{i+1}
\binom {i} {k-1}
y^{ k} .
\end{align}
Now $\sum_{k = 1}^{i+1}
\binom {i} {k-1}
y^{ k} =
y \sum_{k = 0}^i  \binom {i} {k} y^{ k} = y(1+y)^i $.
Thus using the fact that $x \le 1$, $\frac{1}{2\sigma^2+1} \le \frac{1}{2\sigma^2}$ and $y = 1/(2\sigma^2)$, we get
\begin{align}
 \|\hat \rho_{i+1} - \hat \rho_i\|_1 
&\le  
\frac{1}{1+2\sigma^2} + \frac{1}{2\sigma^2}\left( 1 + \frac{1}{2\sigma^2}\right)^i,
\end{align}
and the result follows.
\end{proof}
Clearly then by the telescoping sum, the trace distance between any pair of encrypted diagonal states can be easily bounded.
\begin{lemma} 
\label{lem:any-diagonal-trace-distance}
Let $n$ be any positive integer, and let $i$ and $j$ be non-negative integers such that $i < j \le n$. Then for $\sigma  > 0$, the trace distance between $\hat \rho_{i}$ and $\hat \rho_j$ is at most 
\begin{align}
\frac{1}{2} \|\hat \rho_{i } - \hat \rho_j\|_1 
\le
\frac{j}{4\sigma^2} + \frac{1}{2\sigma^2} 
\left( 1 + \frac{1}{2\sigma^2}\right)^j .
\end{align}
\end{lemma}
\begin{proof}
One just needs to write $
(\hat \rho_{i } -  \hat \rho_{i+1} )
+ \dots  + 
( \hat \rho_{j-1} - \hat \rho_j) .$
There are at most $n$ such bracketed terms, so using the triangle inequality with Lemma \ref{lem:diagonal-trace-distance} gives 
\begin{align}
 \|\hat \rho_{i } - \hat \rho_j\|_1 
&=
\| (\hat \rho_{i } -  \hat \rho_{i+1} )
+ \dots  + 
( \hat \rho_{j-1} - \hat \rho_j) \|_1\notag\\
&\le
\| \hat \rho_{i } -  \hat \rho_{i+1} \|_1
+ \dots  + 
\| \hat \rho_{j-1} - \hat \rho_j) \|_1\notag\\ 
&\le
\frac{j-i}{1+2\sigma^2} + \frac{1}{2\sigma^2}
\sum_{k=i}^{j-1}
\left( 1 + \frac{1}{2\sigma^2}\right)^k
\notag\\
&\le
\frac{j}{1+2\sigma^2} + \frac{1}{\sigma^2} 
\left( 1 + \frac{1}{2\sigma^2}\right)^j 
\notag\\
&\le
\frac{j}{2\sigma^2} + \frac{1}{\sigma^2} 
\left( 1 + \frac{1}{2\sigma^2}\right)^j .
\end{align}
This proves the result.
\end{proof}
We now proceed to obtain a bound on the off-diagonal matrix elements $\rho_{i,j}$. Without loss of generality, assume that $j=i+k$ for $k \ge 0$. 
To analyze this case, we first consider the following technical lemma that is easy to verify.
\begin{lemma}
\label{lem:binomial-upper-bound}
Let $a,i,i_2$ and $k$ be non-negative integers, and let $i_2\le a,i$. Then 
\begin{align}
\binom a {i_2}
\binom i {i_2}
\le
\binom {a+k} {i_2 + k}
\binom {i+k} {i_2 + k}  .
\end{align}
\end{lemma}
\begin{proof}
It is easy to see that $
\binom a k 
\binom i k 
=
\binom {a+k} k 
\binom {i+k} k 
\frac{ \binom {i_2 +k}{k}^2 } { \binom{a+k}{k}  \binom{i+k}{k}  }.$
Next observe that since $i_2\le a$ and $i_2\le i$, 
we have $\frac{ \binom {i_2 +k}{k}^2 } { \binom{a+k}{k}  \binom{i+k}{k}  }\le 1$.
\end{proof}
Using Lemma \ref{lem:binomial-upper-bound} we can arrive derive bounds for the off-diagonal matrix elements $\hat \rho_{i,j}$.
\begin{lemma}
\label{lem:off-diagonal-bound}
Let $i,k$ be non-negative integers and let $j=i+k$. 
Then for $\sigma>0$,
\begin{align}
\sum_{a \ge 0 }
|\<a| \hat \rho_{i,j} |a+k\> | \le  \left(
\frac{1+2\sigma^2}{4\sigma^4}
\right)^k.
\end{align}
\end{lemma}
\begin{proof}
Using the exact form for the matrix element $\<a| \hat \rho_{i,j} |a+k\>$ as given in Lemma \ref{lem:exact-form}
and Lemma \ref{lem:binomial-upper-bound} to bound the binomial coefficients therein, 
we get
\begin{align}
|\<a| \hat \rho_{i,j} |a+k\> | \le 
\frac{1}{1+2\sigma^2}
\sum_{i_2 = 0, \dots , {\rm min}(a,i)}  
\binom {a+k} {i_2+k} 
\binom {i+k} {i_2+k} 
y^{2i_2+k} x^{a+i+k} ,
\end{align} 
where $y = 1/(2\sigma^2)$ and $x = (2\sigma^2)/ (1+2\sigma^2)$.
Using Lemma \ref{lem:exact-form} again, we get
\begin{align}
|\<a| \hat \rho_{i,j} |a+k\> | \le (y/x)^k \<a +k | \hat \rho_{i+k} |a+k\>.
\end{align}
Using the fact that $ \hat \rho_{i+k}$ has unit trace, we easily get
$\sum_{a \ge 0 }|\<a| \hat \rho_{i,j} |a+k\> | \le  (y/x)^k = \left(
\frac{1+2\sigma^2}{4\sigma^4}
\right)^k$.
\end{proof}
We are now ready to prove the main result.
\begin{proof}[Proof of Theorem \ref{thm:mainresultofpaper}]
First we prove that without loss of generality, we can let the any two input states to our scheme $\rho$ and $\rho'$ be pure states. 
Now consider the case where $\rho$ and $\rho'$ are mixed states. Then both of these states can always be written as 
\begin{align}
\rho = \sum_{j\ge 1} p_j |\phi_j\>\<\phi_j|,\quad
\rho' = \sum_{j\ge 1} p'_j |\phi'_j\>\<\phi'_j|,
\end{align}
such that $p_j = p_j'$ for every $j\ge 1$. In this decomposition, the states 
$|\phi_j\>$ and $|\phi_k\>$ need not be distinct even when $j \neq k$. 
Similarly, 
$|\phi'_j\>$ and $|\phi'_k\>$ need not be distinct even when $j \neq k$. 
Here, we must have $p_j$ to be non-negative and $\sum_{j \ge 1} p_j = 1.$
Then we use the linearity of the quantum channel $\mathcal E$ to see that
\begin{align}
 \mathcal E (\rho ) - \mathcal E(\rho') 
 = 
\sum_{j\ge 1} p_j \left( \mathcal E (  |\phi_j\>\<\phi_j| ) - \mathcal E( |\phi'_j\>\<\phi'_j| ) \right).
\end{align}
Applying the triangle inequality for the trace norm, we get
\begin{align}
\| \mathcal E (\rho ) - \mathcal E(\rho') \|_1
 \le
\sum_{j\ge 1} p_j \left\| \mathcal E (  |\phi_j\>\<\phi_j| ) - \mathcal E( |\phi'_j\>\<\phi'_j| ) \right\|_1.
\end{align}
It hence follows that 
\begin{align}
\| \mathcal E (\rho ) - \mathcal E(\rho') \|_1
 \le
\max_{j\ge 1} \left\| \mathcal E (  |\phi_j\>\<\phi_j| ) - \mathcal E( |\phi'_j\>\<\phi'_j| ) \right\|_1. \label{eq:maximization-over-pure-states}
\end{align}
From \eqref{eq:maximization-over-pure-states}, we can see that we can maximize over the trace distance between encrypted pure states to maximize $\| \mathcal E (\rho ) - \mathcal E(\rho') \|_1$.
It thus suffices to consider $\rho=|\phi\>\<\phi|$ and $\rho'=|\phi'\>\<\phi'|$ to be pure states in this security proof,
where $|\phi\>= \sum_{i\ge 0} \lambda_i |i\>$ and 
$|\phi'\>= \sum_{i\ge 0} \mu_i |i\>$.
We make this assumption with loss of generality in the remainder of this proof.

Consider the matrix
\begin{align}
T = \mathcal E (\rho ) - \mathcal E(\rho') 
=  
\sum_{a,b\ge 0}
\sum_{i=0}^n \sum_{j=0}^n
\left(
\lambda_i \lambda_j^* -\mu_i \mu_j^* 
\right)
I_{a,b,i,j}|a\>\<b|.\label{eq:diff-matrix}
\end{align}
Now let $D$ be the diagonal component of $T$ and $O$ be the off-diagonal component of $T$. 
By the triangle inequality, we will have $\| T \|_1  = \| D+O\|_1  \le \| D\|_1 + \| O \|_1$.
We now proceed to bound the diagonal component. 

By definition, we have
\begin{align}
D = \sum_{a \ge 0 }\sum_{i,j=0}^n \left(
\lambda_i \lambda_j^* -\mu_i \mu_j^* 
\right)I_{a,a,i,j} |a\>\<b|.
\end{align}
Using Lemma \ref{lem:exact-form}, we know that $I_{a,b,i,j} = 0$ whenever $i \neq j$. Therefore, the above expression for $D$ simplifies to yield
\begin{align}
D = \sum_{a \ge 0 }\sum_{i=0}^n \left(
|\lambda_i|^2  -|\mu_i|^2
\right)I_{a,a,i,j} |a\>\<b|
=
\sum_{a \ge 0 }\sum_{i=0}^n \left(
|\lambda_i|^2  -|\mu_i|^2
\right) \<a| {\hat \rho}_{i}|a\>.
\end{align}
Hence it follows that
\begin{align}
\|D\|_1 
=
\| \mathcal E(\omega) - \mathcal E(\omega')\|_1,
\end{align}
where 
\begin{align}
\omega = \sum_{i=0}^n |\lambda_i|^2 |i\>\<i|,\quad
\omega' = \sum_{i=0}^n |\mu_i|^2 |i\>\<i|.
\end{align}
Since $\omega$ and $\omega'$ are mixed states that are diagonal in the Fock basis,
we can use \eqref{eq:maximization-over-pure-states} to see that
\begin{align}
\|D\|_1 
\le \max_{0\le i<j\le n} \| \mathcal E(|i\>\<i|) - \mathcal E(|j\>\<j|)\|_1
= \| \hat \rho_i - \hat \rho_j \|_1.
\end{align}
Using Lemma \ref{lem:any-diagonal-trace-distance}, 
\begin{align}
\|D\|_1 \le 
\frac{n}{2\sigma^2} + \frac{1}{\sigma^2} 
\left( 1 + \frac{1}{2\sigma^2}\right)^n.
\end{align}

For the off-diagonal elements we can use the Gersgorin Circle Theorem (GCT).
First, note that for any $i$ and $j$, 
$|\lambda_i \lambda_j^* -\mu_i \mu_j^*|\le 2$.
From the GCT the 1-norm of $O$ is at most the sum of the absolute values of all of its matrix elements. Notice that 
\begin{align}
 O 
&=
 \sum_{a\ge 0, k \ge 1}
\sum_{i=0}^n \sum_{j=0}^n
\left(
\lambda_i \lambda_j^* -\mu_i \mu_j^* 
\right)
I_{a,a+k,i,j}|a\>\<a+k|
+
\sum_{a\ge 0, k \ge 1}
\sum_{i=0}^n \sum_{j=0}^n
\left(
\lambda_i \lambda_j^* -\mu_i \mu_j^* 
\right)
I_{a+k,a,i,j}|a+k\>\<a|
\end{align}
Hence we obtain from the GCT that
\begin{align}
\|O\|_1
&\le
\sum_{a\ge 0, k \ge 1}
\left| \sum_{i=0}^n \sum_{j=0}^n
\left(
\lambda_i \lambda_j^* -\mu_i \mu_j^* 
\right)
I_{a,a+k,i,j}
\right|
+
\sum_{a\ge 0, k \ge 1}
\left|
\sum_{i=0}^n \sum_{j=0}^n
\left(
\lambda_i \lambda_j^* -\mu_i \mu_j^* 
\right)
I_{a+k,a,i,j}
\right|
\notag\\
&\le
2 \sum_{a\ge 0, k \ge 1}
\sum_{i=0}^n \sum_{j=0}^n
\left( 
|I_{a,a+k,i,j}| + |I_{a+k,a,i,j}|
\right),
\end{align}
where we have used the triangle inequality in the second inequality above.
Using the fact that $I_{a,a+k,i,j}$ is only non-zero when $j-i=k$, and similarly for $I_{a+k,a,i,j}$, we get 
\begin{align}
\|O\|_1
&\le
2 \sum_{a\ge 0, k \ge 1}
\sum_{i=0}^n 
\left( 
|I_{a,a+k,i,i+k} |+ |I_{a+k,a,i+k,i}|
\right).
\end{align}
Combining this with the fact that 
$I_{a,a+k,i,i+k} = I_{a+k,a,i+k,i}^*$ we get
\begin{align}
\|O\|_1
&\le
4 \sum_{a\ge 0, k \ge 1}
\sum_{i=0}^n 
\left| 
I_{a,a+k,i,i+k}
\right|.
\end{align}
Using Lemma \ref{lem:off-diagonal-bound} with $\sigma^2 > 0$ for the geometric sum,
this becomes
$\|O\|_1 \le 4 \sum_{i=0}^n \sum_{k \ge 1 } \sigma^{-2k} \le 
\frac{4(n+1) \sigma^{-2} }{1-\sigma^{-2}} 
\le 
8(n+1) \sigma^{-2}$.
The result then follows.
\end{proof}

\section{Multi-mode security}
The security of our scheme on multiple modes arises from a telescoping sum on the modes, when the multi-mode state is a separable state. 
To see this explicitly, 
let the encryption operator on $m$ modes be $\mathcal E$. 
Because the random displacements are chosen independently for every mode, 
we have 
\[{\mathcal E} = {\mathcal E_1} \otimes  \dots\otimes   {\mathcal E_m},\]
where $\mathcal E_j$ denotes an encryption operator on the $j$th mode.
Now let $\mathcal I$ denote the identity channel on a single mode.
Then for any two $m$-mode states $\rho = \rho_1 \otimes \dots \otimes \rho_m$ and $\tau_1 \otimes \dots \otimes \tau_m$ with a tensor product structure, we can write
\begin{align}
&(\mathcal E_1 \otimes  \dots\otimes   \mathcal E_m)(\rho)
-
(\mathcal E_1 \otimes  \dots\otimes  \mathcal E_m)(\tau)\notag\\
&\mathcal E_1(\rho_1) \otimes  \dots\otimes   \mathcal E_m(\rho_m)
-
\mathcal E_1(\tau_1) \otimes  \dots\otimes  \mathcal E_m(\tau_m)\notag\\
=&
A_1 \otimes \dots \otimes A_m - B_1 \otimes \dots \otimes B_m  
\end{align}
where $A_j = \mathcal E_j(\rho_j)$
and
$B_j = \mathcal E_j(\tau_j)$.
Using the telescoping sum, we have
\begin{align}
A_1 \otimes \dots \otimes A_m - B_1 \otimes \dots \otimes B_m 
=&  (A_1 \otimes  A_2 \otimes \dots \otimes A_m 
- B_1 \otimes A_2 \otimes \dots \otimes A_m) \notag\\
&+ 
(B_1\otimes A_2 \otimes A_3 \dots  \otimes A_m 
- B_1 \otimes B_2\otimes  A_3 \otimes \dots \otimes A_m) \notag\\
&+
\dots \notag\\
&+
(B_1 \otimes \dots \otimes B_{m-1} \otimes A_m 
- B_1 \otimes \dots \otimes B_m).
\end{align}
By applying the triangle inequality for the trace norm of each of the above bracketed terms, then we get
\begin{align}
\| A_1 \otimes \dots \otimes A_m - B_1 \otimes \dots \otimes B_m \|_1
\le &  \|(A_1 \otimes  A_2 \otimes \dots \otimes A_m 
- B_1 \otimes A_2 \otimes \dots \otimes A_m) \|_1 \notag\\
&+ 
\| (B_1\otimes A_2 \otimes A_3 \dots  \otimes A_m 
- B_1 \otimes B_2\otimes  A_3 \otimes \dots \otimes A_m)\|_1 \notag\\
&+
\dots \notag\\
&+
\| (B_1 \otimes \dots \otimes B_{m-1} \otimes A_m 
- B_1 \otimes \dots \otimes B_m)\|_1.
\end{align}
Using the multiplicativity of the trace norm under the 
tensor product and the fact that every quantum state has a trace norm equal to one so that $\|A_j\|_1 = \|B_j\|_1$,
we find that 
\begin{align}
\| A_1 \otimes \dots \otimes A_m - B_1 \otimes \dots \otimes B_m \|_1
\le &  \sum_{j=1}^m \|A_j - B_j\|_1 .
\end{align}
If every single mode state $\rho_j$ and $\tau_j$ have at most $n$ photons, and every mode is randomly displaced independently with displacement vector taken from a complex Gaussian distribution of standard deviation $\sigma$ and mean 0, 
using the above inequality, 
we can see that the trace distance between the encrypted states 
$A_1\otimes \dots \otimes A_m$ and 
$B_1\otimes \dots \otimes B_m$ 
is simply at most $m$ times of the trace distance between arbitrary displacement-encrypted single mode states with at most $n$ photons.

\section{Bounded photon number and bounded energy}
In this section, we prove that a quantum state with bounded expected energy is well approximated by a quantum state with a bounded number of photons. 
Using $\hat n = \sum_{n\ge 0} n|n\>\<n|$ to denote the number operator, 
the expected energy of an arbitrary state $\rho$ supported on the Fock basis is
defined to be
\begin{align}
\tr ( \rho \hat n  )  = \sum_{n \ge 0}n  \< n | \rho|n \>.
\end{align}
First we prove a lemma that reduces the problem of bounding the trace-norm of the difference of density matrices to evaluating bounds on Euclidean norms, which is reminiscent of \cite[Lemma 2]{ouyang2019quantum}.
\begin{lemma} \label{lem:tracenorm-to-euclideannorm}
Let $\rho = \sum_{j} p_j |\phi_j\>\<\phi_j|$ and $\rho' = \sum_{j} p_j |\phi'_j\>\<\phi'_j|$ be density operators. Furthermore, let 
$|\epsilon_j\> =  |\phi'_j\> - |\phi_j \>$.
Then 
\begin{align}
\|\rho - \rho'\|_1 \le \sum_{j} p_j \left(  2\sqrt{\<\epsilon_j|\epsilon_j\>} +  \<\epsilon_j|\epsilon_j\> \right).
\end{align}
\end{lemma}
\begin{proof}
By definition, $|\phi'_j\> = |\phi_j\> + |\epsilon_j\>$. Hence,
\begin{align}
|\phi'_j\>\<\phi'_j|  - |\phi_j\>\<\phi_j|
=|\phi_j\>\<\epsilon_j|  + |\epsilon_j\>\<\phi_j| + |\epsilon_j\>\<\epsilon_j|.
\end{align}
By the triangle inequality, we have
\begin{align}
\| |\phi_j\>\<\epsilon_j|  + |\epsilon_j\>\<\phi_j| + |\epsilon_j\>\<\epsilon_j| \|_1
\le
\| |\phi_j\>\<\epsilon_j| \|_1 + \| |\epsilon_j\>\<\phi_j| \|_1
 + \| |\epsilon_j\>\<\epsilon_j| \|_1.
\end{align}
Now we use the definition of the trace norm where for any operator $A$, we have $\|A \|_1 = \max_U \{\tr(AU): \|U\|_\infty \le 1\}$ where $\|U\|_\infty$ denotes the maximum singular value of $U$. For Hermitian operators, it suffices to consider the maximization of $U$ over unitary operators. For $\|  |\epsilon_j\>\<\epsilon_j| \|_1 $, 
the optimal $U$ is just the identity operator, and hence 
$\| |\epsilon_j\>|<\epsilon_j| \|_1 = \<\epsilon_j|\epsilon_j\>$.
To evaluate $\| |\epsilon_j\>\<\phi_j| \|_1$, consider the unitary operator $U$ that swaps the normalized states 
$|\epsilon_j\> / \sqrt{\<\epsilon_j |\epsilon_j \>}$
and 
$|\phi_j\>$.
Then we can see $\| |\epsilon_j\>\<\phi_j| \|_1  =  \<\epsilon_j|\epsilon_j\> / \sqrt{ \<\epsilon_j|\epsilon_j\>}$.
A similar argument applies for evaluating $\| |\phi_j\>\<\epsilon_j| \|_1$, and the result follows.
\end{proof}
Now we present a lemma regarding approximating a pure state with another pure state that has at most $n$ photons.
\begin{lemma} \label{lem:bounded-support-approximant}
Let $|\phi\>$ be any pure state supported on the Fock basis that has expected energy $\mu$. Let $n$ be an integer such that $n \ge \mu$.
Then there exists a state $|\psi\>$ that has at most $n$ photons and 
$\| |\phi\> - |\psi\> \| \le 2 \sqrt{\mu/n}.$
\end{lemma}
\begin{proof}
In the Fock basis, we have
$|\phi\> = \sum_{j=0}^\infty \lambda_j |j\>$ where $\lambda_j$ are complex numbers and $\sum_{j=0}^\infty |\lambda_j|^2 = 1$.
Now consider $|\psi\> = \sum_{j=0}^n \lambda_j |j\> + x |n+1\>$ for some complex number $x$ such that 
$|x|^2 = 1-\sum_{j=0}^n |\lambda_j|^2$.
Hence it follows from the triangle inequality that
\begin{align}
\||\psi\> - |\phi\> \|
=
\||\epsilon\> - x|n+1\> \|
\le 
\||\epsilon\> \| + \| x|n+1\> \|,
\end{align}
where $|\epsilon\> = \sum_{j=n+1}^\infty \lambda_j |j\>$.
Since $\| x|n+1\> \| = |x|$ and 
\begin{align}
\||\epsilon\> \| = 
\sqrt{ \sum_{j \ge n+1}^\infty |\lambda_j|^2 }
=  \sqrt{|x|^2},
\end{align}
we get
\begin{align}
\||\psi\> - |\phi\> \|
\le  2|x| .
\end{align}
Now let $X$ denote a random variable that is equal to $j$ with probability $|\lambda_j|^2$. It then follows that 
\begin{align}
\Pr[X \ge n+1] = |x|^2.
\end{align}
By definition of the expected energy, we know that
\begin{align}
\mathbb E(X) = \sum_{j=0}^\infty j |\lambda_j|^2 = \mu.
\end{align}
Since $X$ is a non-negative random variable, we can use Markov's inequality, so that for any real number $b$ such that $b \ge 1$, we have
\begin{align}
\Pr[X \ge b \mu] \le \frac{1}{b}.
\end{align}
Since $n\ge \mu$, we have 
\begin{align}
\Pr[X \ge n+1] \le \frac{\mu}{n}.
\end{align}
Hence it follows that 
\begin{align}
\||\psi\> - |\phi\> \|
\le
2\sqrt{\mu/n}.
\end{align}
\end{proof}

\end{widetext}

\end{document}